\documentclass[11pt,a4paper]{extarticle}%
\usepackage[margin=1in]{geometry}

\usepackage[margin=1in]{geometry}
\usepackage{amsmath}%
\usepackage{amsfonts}%
\usepackage{amssymb}%
\usepackage{amsthm}%

\usepackage{graphicx}%

\usepackage{hyperref}

\theoremstyle{definition}

\theoremstyle{plain}
\newtheorem{lemma}{Lemma}[section]

\newtheorem{theorem}{Theorem}[section]

\newtheorem{invariant}{Invariant}

\newcommand{\bt}[1] {BT_{#1}}
	
	\newcommand{\getedge}[2]{getEdge_{#1}(#2)}
	\def\compactify{\itemsep=0pt \topsep=0pt \partopsep=0pt \parsep=0pt}
\usepackage{color}
	\newcommand{\pc}[1] {pc_{#1}}
	
\newcommand{\Patrascu}{P\v{a}tra\c{s}cu}

\begin{document}

\title{Improved Worst-Case Deterministic Parallel\\ Dynamic Minimum Spanning Forest}
  \author{
    Tsvi Kopelowitz\\
    Bar-Ilan University\\
    \texttt{kopelot@gmail.com}
    \and Ely Porat\\
    Bar-Ilan University\\
    \texttt{porately@cs.biu.ac.il}
    \and Yair Rosenmutter\\
    Bar-Ilan University\\
    \texttt{yapoke@gmail.com}
  }
  \maketitle

  \begin{abstract}

    This paper gives a new deterministic algorithm for the dynamic Minimum Spanning Forest (MSF) problem in the EREW PRAM model, where the goal is to maintain a MSF of a weighted graph with $n$ vertices and $m$ edges while supporting edge insertions and deletions.
    We show that one can solve the dynamic MSF problem using $O(\sqrt n)$ processors and $O(\log n)$ worst-case update time, for a total of $O(\sqrt n \log n)$ work. This improves on  the work of Ferragina [IPPS 1995] which costs $O(\log n)$ worst-case update time and $O(n^{2/3} \log{\frac{m}{n}})$ work.
  \end{abstract}

\section{Introduction}
In the \emph{dynamic minimum spanning forest} (MSF) problem, the goal is maintain a MSF $F$ of an undirected dynamic graph $G=(V,E)$ with weight function $w:E\rightarrow \mathbb{R}$, while supporting edge insertions and deletions. The dynamic MSF problem is one of the most fundamental dynamic graph problems, and has been used as a subroutine for solving many other graph problems~(\cite{DBLP:conf/focs/AbrahamDKKP16},\cite{DBLP:conf/focs/EppsteinGIN92},\cite{DBLP:conf/swat/Thorup00},\cite{DBLP:journals/combinatorica/Thorup07}).
The first sequential algorithm solving the dynamic MSF problem has a worst-case update time of $O(\sqrt m)$, where $m$ is the number of edges, and was introduced by Frederickson~\cite{Frederickson85}.
Using the sparsification technique of Eppstein et al.~\cite{DBLP:conf/focs/EppsteinGIN92,EppsteinGIN97} on Fredrickson's algorithm reduces the worst-case update time to $O(\sqrt n)$, where $n$ is the number of vertices.
Both of these results are deterministic.
While there have been several improvements on the time cost when using randomization or allowing amortization, the $O(\sqrt n)$ time bound is the best known for deterministic worst-case dynamic MSF.

\paragraph{Dynamic MSF in the PRAM model.}
While the dynamic MSF problem in sequential models has received a lot of attention from researchers, there has been no progress since the 90s in the PRAM Model.
Das and Ferragina~\cite{DasF94} presented a dynamic MSF algorithm in the EREW PRAM model, which is based on Frederickson's~\cite{Frederickson85} sequential algorithm, that uses $O(\frac{m^{2/3}}{\log n})$ processors, $O(\log n)$ worst-case time, and $O(m^{2/3})$ work.
Ferragina~\cite{Ferragina95} showed how to parallelize the sparsification technique of Eppstein et al.~(\cite{DBLP:conf/focs/EppsteinGIN92}, \cite{EppsteinGIN97}), thereby obtaining a fully dynamic MSF algorithm in the EREW PRAM model, that uses $O\left(\frac{n^{2/3}\log \frac m n}{\log n}\right)$ processors, $O(\log n)$ worst-case time, and $O(n^{2/3} \log{\frac{m}{n}})$ work.
Liang and McKay \cite{LiangMckay} proposed a different parallel algorithm for dynamic MSF that uses $O(n^{2/3})$ processors and has $O(\log{n} \cdot \log{\frac{m}{n}})$ parallel worst-case time.

\paragraph{Our results.}
In this paper we give the first improvement on dynamic MSF in the EREW PRAM model in over 20 years (in terms of deterministic worst-case update times). The main result is summarized by the following theorem.
\begin{theorem}
	\label{theorem:mstpar}
	There exists a deterministic algorithm for the dynamic MSF problem in the EREW PRAM model that uses $O(\sqrt{n})$ processors and has a parallel worst-case update time of $O(\log{n})$. The resulted work of the algorithm is $O(\sqrt{n} \log{n})$.
\end{theorem}

\paragraph{Dynamic MSF and Dynamic Connectivity.}
In the \emph{dynamic connectivity} problem the goal is to maintain a dynamic graph $G$ with edge insertions and deletions, while supporting connectivity queries: ``given two vertices in $G$, are the vertices in the same connected component?".
The dynamic connectivity problem is a weaker version of the dynamic MSF problem, since one way of solving dynamic connectivity is to maintain a spanning forest of $G$ and using dynamic connectivity data structures for forests such as ~\cite{DBLP:journals/jcss/SleatorT83}.
Thus, Frederickson's algorithm together with the sparsification techinque yield a $O(\sqrt n)$ worst-case deterministic update time for dynamic connectivity. Recently, Kejlberg-Rasmussen et al.~\cite{Kejlberg-Rasmussen16} reduced the runtime slightly to $O\left(\sqrt{\frac{n(\log{\log{n}})^2}{\log n}}\right)$. The proof of Theorem~\ref{theorem:mstpar} is based on the approach used by Kejlberg-Rasmussen et al.~\cite{Kejlberg-Rasmussen16} for solving dynamic connectivity.

\subsection{Algorithmic Overview}
Throughout the paper we apply the standard assumption that the graph is \emph{sparse}, i.e., the graph has $m=O(n)$ edges.
In the sequential case this assumption is permissable due to the sparsification technique of~\cite{EppsteinGIN97}.
We later show how to extend the sparsification technique for dynamic MSF to the EREW PRAM model.
We also assume throughout the paper that the maximum degree in $G$ is 3 by applying the techniques of Frederickson~\cite{Frederickson85}.
This last assumption costs an $O(1)$ worst-case time additive overhead per operation.

We are now ready to provide an overview of our techniques.
We emphasize that our overview sacrifices accuracy for the sake of intuition.
An accurate description of the techniques that we use is given in the rest of the paper.
We believe it is best to first discuss a sequential version of our dynamic MSF algorithm, which is
an $O(\sqrt {\log n})$ factor slower than the algorithm of Frederickson~\cite{Frederickson85}
after applying the sparsification technique of~\cite{EppsteinGIN97}.
Nevertheless, the proof of the following theorem is helpful for understanding our proof of Theorem~\ref{theorem:mstpar}.

\begin{theorem}
	\label{theorem:mstseq}
	There exists a sequential deterministic algorithm for the dynamic MSF problem, which has a worst-case update time of $O(\sqrt{n \log n})$.
\end{theorem}

\paragraph{Euler tours and lists.}
The proof of Theorem~\ref{theorem:mstseq} is based on the dynamic connectivity algorithm presented in \cite{Kejlberg-Rasmussen16}.
The basic technique is to maintain \emph{Euler tours} of the trees in $F$, which is the MSF of $G$. An \emph{Eulerian circuit} in a directed graph is a tour on the edges of the graph, starting and ending on the same vertex, in which each edge is visited exactly once. There is an Eulerian circuit if and only if for every vertex $u$ in the graph the out-degree of $u$ is the same as the in-degree of $u$.
For a tree $T$ in a spanning forest of an undirected graph $G$, an Euler tour of $T$ is a list of the edges in $T$. The Euler tour of $T$ is created by treating each undirected edge as two directed edges in different directions, so that for every vertex $u$ in $T$ the out-degree of $u$ is the same as the in-degree of $u$.
The authors of~\cite{Kejlberg-Rasmussen16} showed how to reduce the problem of maintaining Euler tours to that of supporting splits and merges of linked lists, and finding a \emph{minimum weight replacement} (MWR) edge in the case of deleting an edge in $G$ that is also in $F$.
While supporting operations on lists is generally straightforward, being able to find a MWR edge is the challenging aspect.
Nevertheless, lists turns out to be very convenient for parallelization since different processors can focus on different parts of the list.

\paragraph{Chunks and LSDS.}
In order to be able to manipulate the lists efficiently, we partition each list into \emph{chunks} of size $K=O(\sqrt {n\log n})$, so that there are $J=O(\sqrt {n/ \log n})$ chunks.
The lists contain copies of vertices as they appear in the Euler tour, and so, by our assumptions, every chunk $c$ has as most $O(K)$ edges touching vertices with copies in $c$.
Throughout the execution of the algorithm chunks are merged and split, either due to a list splitting at a chunk or in order to guarantee that every chunk is of size at most $O(K)$.
Above the chunks, we use a tree-like data structure, called the \emph{list sum data structure} (LSDS) that has at most $J$ leaves (one for each chunk) and height $O(\log J)$.

The basic idea is to separately aggregate useful information for finding MWR edges from all of the elements in each chunk, and then use the LSDS to aggregate all of the information from all of the chunks.
The information stored for a chunk $c$ is an array $CAdj_c$ of size $J$ that stores the connectivity information between $c$ and every other chunk.
The array stores one entry for each chunk such that the $i$'th entry in all of these arrays represents the chunk with id $i$.
This information stored in $CAdj_c$ is derived from the at most $O(K)$ edges touching vertices that have a copy in $c$.
The information stored at an LSDS tree vertex $u$ is an array $CAdj_u$ of size $J$ which stores the aggregate of all of the chunk information of chunks in the subtree of $u$.

When a chunk $c$ is the outcome of either a split or merge of chunks, $CAdj_c$ is updated by scanning the elements of $c$ in $O(K)$ worst-case time (since each chunk contains at most $O(K)$ elements).
Then, for every other chunk $c'$ the connectivity information between $c'$ and $c$ that needs to be stored in $CAdj_{c'}$ is derived from $CAdj_c$.
Since there are $O(J)$ chunks, this takes $O(J)$ worst-case time.
The information in $CAdj_c$ is then propagated up the LSDS tree to the ancestors of $c$, spending $O(J)$ worst-case time per ancestor for a total of $O(J \log J) = O(K)$ worst-case time.
Finally, the entries corresponding to $c$ in \emph{all} of the vertices of the LSDS tree are updated by scanning the LSDS tree while only accessing the entries in the arrays that correspond to $c$.
This last part spends $O(1)$ worst-case time per tree vertex for a total of $O(J)$ worst-case time.
Thus, the cost of merging and splitting chunks ends up being $O(J+K) = O(\sqrt{n\log n})$.

In addition to the splitting and merging of chunks, the algorithm will sometimes need to split or merge LSDS structures.
These splits and merges are standard tree operations, and each such tree operation touches $O(\log n)$ tree-vertices.
The time cost is dominated by updating the $CAdj$ arrays (of size $O(J)$ each) of the vertices touched during the tree operations, for a total of $O(\sqrt{n\log n})$ worst-case time.

\paragraph{Finding a MWR edge.}
The method for finding a MWR edge is to find the lightest edge connecting the vertices with copies in one list $L_1$ and the vertices with copies in a second list $L_2$.
We remark that there are some crucial details of this process which we skip in the overview description here.

The algorithm uses the $CAdj$ array stored at the root of the LSDS for $L_2$, which specifically contains the connectivity information between the chunks in $L_2$ and all other chunks.
The MWR edge is found by scanning the $O(J)$ chunks of $L_1$, and for each such chunk $c$ we look at the entry  corresponding to $c$ in the $CAdj$ array of the root of $L_2$.
The chunk in $L_1$ that has the smallest entry in the $CAdj$ array contains the lightest edge connecting $L_1$ and $L_2$, and the algorithm scans all of the $O(K)$ edges touching vertices with copies in that chunk in order to find the MWR edge.
The entire process costs $O(J+K) = O(\sqrt {n\log n})$ worst-case time.

\paragraph{Parallel dynamic MSF.}
One of the advantages of our new sequential MSF algorithm is that it leads to an improved parallel dynamic MSF algorithm.
The reason for the time cost being $O(\sqrt {n \log n})$ in the sequential algorithm is due to scanning all of the elements in a chunk, scanning all of the chunks of a list, and scanning all of the vertices in an LSDS.
By utilizing several tournament like trees, we show how all of these tasks can be executed in parallel, with a cost of $O(\log n)$ parallel worst-case time.

\paragraph{Parallel sparsification.}
The sparsification method of Eppstein et al.~\cite{EppsteinGIN97} allows one to reduce the dependency of the sequential time cost for dynamic MSF in terms of the number of edges. In particular, this method admits a conversion of any algorithm for solving dynamic MSF with polynomial sequential time cost $f(m)$ to be reduced to a time cost of $f(n)$ by focusing on the special case in which $m=O(n)$.
Roughly speaking, the method uses a tree based data structure that has $O(\log n)$ levels where the number of edges associated with a tree vertex at level $i$ is $n/2^i$.
Each update necessitates at most one update at each level and so the total time cost becomes $\sum_{i=1}^{O(\log n)} f(n/2^i)$.

Unfortunately, the sequential sparsification does not transfer immediately to the parallel setting.
In the sequential algorithm of Section~\ref{sec:seqmin} the time cost on a graph with $O(n')$ edges is $f(n') = O(\sqrt{n'\log n'})$, and so using the sparsification method we have that $\sum_{i=1}^{O(\log n)} \sqrt{(n/2^i)\log (n/2^i)})= O(\sqrt{n\log n})$.
However, in the parallel algorithm, since $f(n') = O(\log n')$ we have that $\sum_{i=1}^{O(\log n)} \log(n/2^i)= O(\log^2n )$. This adds an $O(\log n)$ factor to the cost.

Ferragina~\cite{Ferragina95} introduced a parallel technique to apply the sparsification data structure of Eppstein et al.~\cite{DBLP:conf/focs/EppsteinGIN92} to any parallel algorithm solving the fully dynamic MSF problem, with a $O(\log{\frac{m}{n}})$ factor to the total work. However, we are interested in a parallel sparsification method that does not increase any of the asymptotic costs of the dynamic MSF data structure. Thus, we introduce an augmentation of the Eppstein et al.~\cite{EppsteinGIN97} sparsification method which achieves this goal. This augmentation requires a more detailed description of how the spasification data structure works, which we discuss in Section~\ref{sec:parallel_sparse}.

\subsection{Other Related Work}
\paragraph{Related work on sequential dynamic MSF.}
In the sequential setting, the fastest deterministic worst-case algorithm for dynamic MSF has an $O(\sqrt n)$ worst-case update time~\cite{Frederickson85,EppsteinGIN97}.
Henzinger and King~\cite{DBLP:conf/icalp/HenzingerK97} showed a deterministic algorithm for dynamic MSF with amortized update time of $O(n^{1/3} \log n)$.
Holm et al.~\cite{DBLP:journals/jacm/HolmLT01} designed another deterministic algorithm which has an amortized update time of $O(\log^4 n)$, and later Holm et al~\cite{DBLP:conf/esa/HolmRW15}
improved the amortized update time to $O(\frac{\log^4 n}{\log{\log n}})$.
When allowing randomization, Wulff-Nilsen~\cite{DBLP:conf/stoc/Wulff-Nilsen17} gave a randomized Las-Vegas algorithm for dynamic MSF with an expected worst-case update time of $O(n^{\frac{1}{2} - \epsilon})$ for some constant $\epsilon > 0$, and Nanongkai et al.~\cite{DBLP:conf/focs/NanongkaiSW17} gave a randomized Las-Vegas algorithm with an expected worst-case update time of $n^{O(\log \log \log n / \log \log n)}$.
\Patrascu{} et al.~\cite{DBLP:journals/siamcomp/PatrascuD06} proved a lower bound of $\Omega(\log n)$ update time for dynamic connectivity, which is also a lower bound for dynamic MSF.

\paragraph{Related work on sequential dynamic connectivity.}
Notice that an algorithm for dynamic MSF is also an algorithm for dynamic connectivity. However, often there are faster dynamic connectivity algorithms.
For the related dynamic connectivity problem, as mentioned above, the fastest deterministic worst-case time algorithm has an $O\left(\sqrt{\frac{n(\log{\log{n}})^2}{\log n}}\right)$ worst-case update time~\cite{Kejlberg-Rasmussen16}.
When allowing randomization, Kapron, King, and Mountjoy~\cite{KapronKM13}
gave a Monte Carlo randomized structure with update time $O(c\log^5 n)$ and one-sided error probability $n^{-c}$.
The update time was later improved to $O(c\log^4 n)$ independently by Gibb et al.~\cite{GibbKKT15} and Wang~\cite{DBLP:journals/corr/Wang15w}.
When allowing amortization
Wulff-Nilsen~\cite{Wulff-Nilsen13a} discovered a \emph{deterministic} data structure with
$O(\log^2 n/\log\log n)$ amortized update time.
Recently, Huang et al.~\cite{DBLP:conf/soda/HuangHKP17} showed that if both randomization and amortization are allowed then dynamic connectivity can be solved in expected $O(\log n(\log\log n)^2)$ amortized time. Recently, Nanongkai and Saranurak~\cite{DBLP:conf/stoc/NanongkaiS17} presented a randomized Las-Vegas algorithm for dynamic MSF with an expected worst-case update time of $O(n^{\frac{1}{2} - \epsilon})$ for some constant $\epsilon > 0$.

\subsection{Organization}
In Section~\ref{sec:seqmin} we give a formal detailed description of the sequential dynamic MSF algorithm proving Theorem~\ref{theorem:mstseq}.
In Section~\ref{sec:basic} we give a formal detailed description of the EREW PRAM dynamic MSF algorithm for sparse graphs, thereby proving Theorem~\ref{theorem:mstpar} for the special case of $m=O(n)$.
Due to space considerations, the description of how to extend the sparsification technique of Eppstein et al.~\cite{EppsteinGIN97} to the EREW PRAM model, without affecting the costs, is deferred to Section~\ref{sec:parallel_sparse}. This technique completes the proof of Theorem~\ref{theorem:mstpar} for general graphs.

\section{Sequential Dynamic MSF - Proof of Theorem~\ref{theorem:mstseq}}
\label{sec:seqmin}
The goal of this section is to present a deterministic sequential algorithm for dynamic MSF in \emph{sparse graphs}, where $m=O(n)$. Theorem~\ref{theorem:mstseq} is proved by applying the sparsification technique of Eppstein et al.~\cite{EppsteinGIN97} to the presented algorithm, allowing the algorithm to also work for general graphs, without changing the time cost.

\subsection{Preliminaries}
Each tree $T$ in the MSF is represented as an Euler tour $Euler(T)$, and each Euler tour (which is a list of edges) is stored as a list of vertices whose order is defined by the Euler tour such that every two adjacent vertices in the list represent an edge in the tour.
We say that an edge in the MSF is a \emph{tree edge}.
When a new edge $e=(u,v)$ is added to the graph, if $e$ connects two different trees in the MSF then $e$ becomes a tree edge, and the two corresponding Euler tours are merged. Otherwise, let $e'$ be the heaviest edge on the path between $u$ and $v$ in the MSF prior to the update. If $w(e)<w(e')$ then the algorithm removes $e'$ from the MSF and adds $e$ instead, thereby rearranging Euler tours.
When a non-tree edge is deleted, there are no changes to any of the Euler tours. However, when a tree edge is deleted, the algorithm looks for a \emph{Minimum Weight Replacement} (MWR) edge to reconnect the MSF, thereby rearranging Euler tours.
Thus, the task of rearranging Euler tours is reduced to \emph{surgical operations} of splitting or merging lists, as expressed in the following lemma (proven in \cite{Kejlberg-Rasmussen16}).

\begin{lemma}[\cite{Kejlberg-Rasmussen16}, Lemma 2.1]
	\label{lemma:surgicaldef}
	Let $G$ be an undirected graph with MSF $F$. Let $T \in F$ be a MST in $F$ that contains a tree edge $e$. If $e$ is deleted, and $T = T_{0} \cup e \cup T_{1}$, then Euler($T_{0}$) and Euler($T_{1}$) are constructed from Euler($T$) with $O(1)$ surgical operations. In the opposite direction, from Euler($T_{0}$) and Euler($T_{1}$), Euler($T_{0} \cup e \cup T_{1}$) is constructed with $O(1)$ surgical operations. It takes $O(1)$ worst-case time to determine which surgical operations to perform.
\end{lemma}

Notice that during the sequence of surgical operations (of splitting or merging lists) that take place during splitting or merging Euler tours (as expressed in Lemma~\ref{lemma:surgicaldef}), the lists of vertices may temporarily not be valid Euler tours.

The discussion above together with Lemma \ref{lemma:surgicaldef} reduces the dynamic MSF problem to the following three subproblems: (1) finding the heaviest edge between two given vertices in a dynamic forest, (2) finding a MWR edge, and (3) implementing surgical operations on lists. The first subproblem is solved by applying the dynamic tree data structure of Sleator and Tarjan~\cite{DBLP:journals/jcss/SleatorT83}, which costs $O(\log{n})$ worst-case time per forest update or path query. The rest of the description focuses on solving the last two subproblems. We emphasize that implementing the surgical operations on lists can be done in a straightforward manner costing $O(1)$ worst-case time per operation, but such an implementation does not seem to support an efficient implementation of finding a MWR edge.

\subsection{The Data Structure}\label{sec:seq_data_struct}
\paragraph{Principal copies.} For a graph vertex $u$ there could be several copies of $u$ in the list of vertices containing $u$  (since the purpose of the list is to represent an Euler tour). The algorithm designates one copy to be the \emph{principal copy} of $u$, denoted by $\pc{u}$, and vertex $u$ stores a bidirectional pointer to $\pc u$.

\paragraph{Chunks.} Each list of vertices is partitioned into consecutive \emph{chunks} of vertices, so that each chunk $c$ contains $O(K)$
vertices and there are $O(K)$ edges incident to vertices whose principal copy is in $c$, for a carefully chosen parameter $K$.
We say that an edge is \emph{adjacent to} or \emph{touches} chunk $c$ if the edge touches a graph vertex whose principal copy is in $c$.
The algorithm maintains Invariant~\ref{edgestochunk} for all chunks (recall that the maximum degree in the graph is 3):

\begin{invariant}
	\label{edgestochunk}
	For chunk $c$ let $n_{c}$ be the sum of the number of vertices in $c$ and the number of edges adjacent to $c$. Then $n_{c} \le 3K$, and if $c$ is not the only chunk in the list containing $c$ then $K\le n_{c}$.
\end{invariant}

Notice that when Invariant \ref{edgestochunk} is violated, a standard technique of merging and/or splitting $O(1)$ adjacent chunks is used in order to restore the invariant. We describe the process for merging and splitting chunks in the proof of Lemma~\ref{lem:chunksplitmerge}.

We assume for now that every list contains at least two chunks. The special case of lists containing only one chunk is addressed in Section~\ref{sec:onechunk}. Thus, we denote the maximum number of possible chunks by $J=O(n/K)$.
Each chunk $c$ is assigned a unique id $id_c\in [J]$, and each vertex in $c$ stores $id_c$.
In order to support quick lookups of chunks, the algorithm stores a $J$ sized array called $chunks$ such that  $chunks[id_c]=c$.
Each chunk $c$ maintains two $J$-length vectors called $CAdj_c$ and $Memb_c$.
$CAdj_c[id_{c'}]$ contains the minimum weight of an edge $(u,v)$ such that $\pc u$ is in $c$ and $\pc v$ is in $c'$. If no such edge exists then we denote $CAdj_c[id_{c'}] = \infty$. In the $Memb_c$ vector, all of the entries are set to 0 except for the entry at $id_c$ which is set to 1.

\paragraph{The LSDS.} For each list $L$, the data structure stores a \emph{list sum data structure} (LSDS) which is implemented as a 2-3 tree whose leaves correspond, in order, to the chunks of $L$. The LSDS supports logarithmic worst-case time inserts, deletes, splits and joins. Each internal vertex $z$ maintains two $J$-length vectors, $Memb_{z}$ and $CAdj_{z}$. $Memb_{z}$ is the entry-wise OR of all the $Memb$ vectors of chunks contained in leaves in the subtree of $z$. $CAdj_{z}$ is the entry-wise minimum of all the $CAdj$ vectors of chunks contained in leaves in the subtree of $z$, as shown in Figure~\ref{fig:lsds}.

In order to efficiently perform surgical operations on lists, we describe an efficient implementation for splitting and merging chunks (Section~\ref{sec:chunk_split_merge}) and an efficient implementation of LSDS operations (Section~\ref{sec:LSDS_oper}). We then use these implementations to show how to efficiently implement the surgical operations and how to find a MWR edge (Section~\ref{sec:surg_oper}).

\begin{figure}
	\centering
	\includegraphics[scale=0.65, trim=0.5cm 10cm 0.5cm 3.5cm]{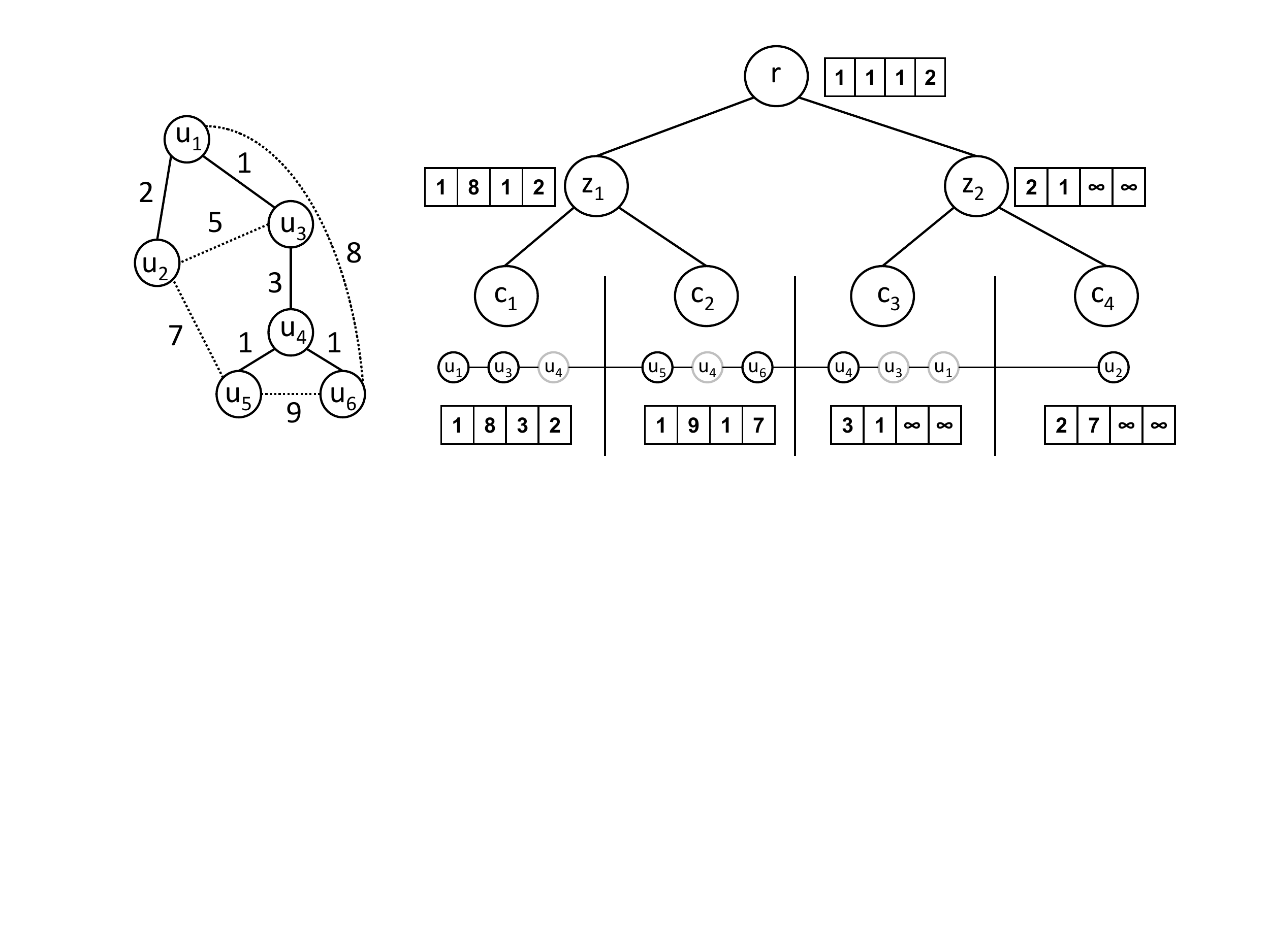}
	\caption{\small
The left part of the figure is a graph where the tree edges are solid edges and the non-tree edges are the dotted edges.
The right part of the figure is a LSDS tree whose leaves are chunks of the list of vertices representing the Euler tour for the graph.
Each chunk contains a sublist of the Euler tour.
The principal copies are the black vertices.
The arrays in each chunk and next to internal tree vertices are the $CAdj$ vectors. }
	\label{fig:lsds}
\end{figure}

\subsection{Splitting and Merging Chunks}\label{sec:chunk_split_merge}
\begin{lemma}
  \label{lem:chunksplitmerge}
  There exists a data structure on chunks that supports splits and merges such that each operation costs $O(J+ K)$ worst-case time.
\end{lemma}

\noindent\begin{proof}
\emph{Splitting.}
Splitting a chunk $c$ can happen for one of two reasons: (1) either the list containing $c$ needs to be split at a given vertex $u$ that is in $c$, or (2) $n_c>3K$ thereby violating Invariant~\ref{edgestochunk}. In the second case, the split location is located by scanning $c$ in $O(K)$ worst-case time. Thus, we assume from now that the algorithm knows the split location $u$.

Splitting the list of vertices in $c$ at vertex $u$ takes $O(1)$ worst-case time.
Let $c_{1}$ and $c_{2}$ be the resulting chunks where $c_1$ contains the first part of the list of vertices from $c$ and $c_2$ contains the second part.
The algorithm sets $id_{c_{1}} = id_{c}$ and allocates a new (unique) id for $c_{2}$.
Next, the algorithm scans all of the vertices in $c_2$ and updates their chunk id.
The new $CAdj$ arrays for $c_{1}$ and $c_2$ are created by iterating over all edges adjacent to $c_{1}$ and $c_2$, respectively, in $O(K)$ worst-case time.
Finally, for each chunk $c'$, the algorithm updates $CAdj_{c'}[id_{c_1}]$ to be $CAdj_{c_1}[id_{c'}]$ and  $CAdj_{c'}[id_{c_2}]$ to be $CAdj_{c_2}[id_{c'}]$, which takes $O(J)$worst-case time.
Thus, the cost for splitting a chunk is $O(J+K)$ worst-case time.

\emph{Merging.}
Merging the lists of vertices of adjacent chunks $c_{1}$ and $c_{2}$ takes $O(1)$ worst-case time. Let $c$ denote the resulting chunk containing the concatenation of the two lists.
The algorithm sets $id_{c} = id_{c_{1}}$, and in $O(K)$ worst-case time the algorithm scans all of the vertices in $c$ in order to update their chunk id.
The new $CAdj$ array for $c$ is created by iterating over all edges incident to $c$ in $O(K)$ worst-case time.
Finally, for each chunk $c'$, the algorithm updates $CAdj_{c'}[id_{c}]$ to be $CAdj_{c}[id_{c'}]$ and sets $CAdj_{c'}[id_{c_2}] = \infty$, which takes $O(J)$ worst-case time.
Thus, the cost for merging two adjacent chunks $O(J+K)$ worst-case time.
\end{proof}

\subsection{Implementing The LSDS Operations}\label{sec:LSDS_oper}
The LSDS supports the following operations:
\begin{itemize}\compactify
	\item \texttt{LSInsert$(c,c')$} - Add a new leaf for chunk $c$ after chunk $c'$.
	\item \texttt{LSDelete$(c)$} - Destroy chunk $c$.
	\item \texttt{LSJoin$(LS_1, LS_2)$} - Concatenate two lists of chunks represented by LSDS $LS_1$ and LSDS $LS_2$.
	\item \texttt{LSSplit$(c)$} - Split the LSDS at chunk $c$.
	\item \texttt{UpdateAdj$(c)$} - (Takes place immediately after an update to $CAdj_c$ and $CAdj_{c'}[id_c]$ for all chunks $c'\neq c$.) Update the $CAdj$ vectors for ancestors of $c$ in the LSDS, and update the $id_c$'th entry of the $CAdj$ vectors for all of the vertices in the LSDS.
\end{itemize}

\begin{lemma}
	\label{lem:lsdsseq}
	There exists an implementation of the LSDS where each of the operations \texttt{LSInsert}, \texttt{LSDelete}, \texttt{LSJoin}, \texttt{LSSplit} and \texttt{UpdateAdj} take $O(J\log{J})$ worst-case time.
\end{lemma}
\begin{proof}

\emph{All operations except for \texttt{UpdateAdj}.}		
Basic tree operations on the LSDS, including access, insertion, deletion, splitting and joining cost $O(\log{J})$ worst-case time each (since each LSDS supports at most $J$ chunks).
Thus, each basic tree operation touches at most $O(\log{J})$ vertices in the tree.
For every vertex $z$ in the LSDS, updating a \emph{single} entry in $CAdj_z$ or $Memb_z$ costs $O(1)$ worst-case time (since the number of children is $O(1)$), and so updating the $CAdj$ and $Memb$ arrays for the $O(\log J)$ vertices touched during a basic tree operation takes at most $O(J\log{J})$ worst-case time.
	
\emph{Operation \texttt{UpdateAdj}.}
Updating the $CAdj$ vectors in the path from the leaf representing $c$ to the root takes $O(J \log{J})$ worst-case time. Updating the $id_c$'th entry of every $CAdj$ array in the tree takes $O(J)$ worst-case time (since the 2-3 tree contains at most $O(J)$ vertices). Thus, the total cost of \texttt{UpdateAdj} is $O(J \log{J})$ worst-case time.
\end{proof}

\subsection{Surgical Operations}\label{sec:surg_oper}

\begin{lemma}
  	\label{lem:surgicalimpl}
	There exists an algorithm in which each surgical operation on lists costs $O(J\log J + K)$ worst-case time and finding a MWR edge costs $O(J+K)$ worst-case time.
\end{lemma}

\noindent\begin{proof}
\noindent\emph{Splitting a list.}
Suppose we split list $L$ at vertex $u$ of chunk $c$ into two parts, $L_{1}$ and $L_{2}$. Let $LS$ be the LSDS representing $L$.
The algorithm splits $c$ at vertex $u$ into $c_{1}$ and $c_{2}$, inserts $c_2$ into $LS$ after the leaf representing $c_1$ and calls \texttt{UpdateAdj}($c_1$) and \texttt{UpdateAdj}($c_2$) in order to update all $CAdj$ vectors in $LS$. Next, the algorithm splits $LS$ into $LS_{1}$ and $LS_{2}$, where the last chunk of $LS_{1}$ is $c_{1}$ and the first chunk of $LS_{2}$ is $c_{2}$. If Invariant \ref{edgestochunk} is violated at either $c_{1}$ or $c_{2}$, then the algorithm executes $O(1)$ splits and merges (followed by $O(1)$ LSDS operations) on $c_{1}$ or $c_{2}$ together with their adjacent chunks in $L_{1}$ or $L_{2}$, respectively, thereby preserving Invariant~\ref{edgestochunk}.  Thus, splitting a list costs $O(J \log{J}+K)$ worst-case time.

\emph{Joining two lists.}
Suppose we join two lists, $L_{1}$ and $L_{2}$, into a single list $L$. Let $LS_{1}$ ($LS_{2}$) be the corresponding LSDS of $L_{1}$ ($L_{2}$). The algorithm calls \texttt{LSJoin}($LS_{1}$, $LS_{2}$) to merge $LS_{1}$ and $LS_{2}$ into $LS$, which costs $O(J \log{J})$ worst-case time.

\emph{Finding a MWR edge.}
Notice that the algorithm looks for a MWR edge between two Euler tours $ET_{1}$ and $ET_{2}$ \emph{only immediately} after splitting Euler tour $ET$ into $ET_{1}$ and $ET_{2}$. Let $LS_{1}$ ($LS_{2}$) be the LSDS corresponding to the list of $ET_{1}$ ($ET_{2}$). Let $r_1$ and $r_2$ be the roots of $LS_1$ and $LS_2$, respectively.

The algorithm constructs an array $\gamma$ of length $J$ in $O(J)$ worst-case time such that if $Memb_{r_2}[i] = 0$ then $\gamma[i]= \infty$, and otherwise $\gamma[i] = CAdj_{r_1}[i]$.
Thus, $\gamma[id_c]< \infty$ if and only if there exists some edges between vertices in $ET_1$ and vertices in chunk $c$ (which must be in $ET_2$).
Moreover, if $\gamma[id_c]< \infty$ then the weight of the minimum weight edge between $ET_1$ and chunk $c$ is $\gamma[id_c]$.
Let ${\hat c}=\arg \min_{\text{chunk } c}\{\gamma[id_c]\}$.
Thus, the minimum weight edge between $ET_1$ and $ET_2$ touches a graph vertex $u$ such that $\pc{u} \in \hat c$.
The algorithm computes $id_{\hat c}$ in $O(J)$ worst-case time by scanning $\gamma$ and looking for the smallest entry.
Then, the algorithm scans all of the $O(K)$ edges touching $\hat c = chunks[id_{\hat c}]$, and for each such edge $e=(u,v)$ where $\pc u \in \hat c$, the algorithm verifies whether the chunk $c_v$ containing $\pc v$ is in $LS_1$ or not by looking at $Memb_{r_1}[id_{c_v}]$. Finally, the algorithm picks the lightest edge that passes the verification.
Thus, the total cost of finding the MWR edge is $O(J+K)$ worst-case time.
\end{proof}

\subsection{Graph Updates}
\begin{proof}[Proof of Theorem~\ref{theorem:mstseq}]
\emph{Edge insertion.}
Suppose we insert a new edge $e=(u,v)$ with weight $w(e)$ to the graph.
Let $c_{1}$ and $c_{2}$ be the chunks containing $\pc u$ and $\pc v$, respectively, and let $LS_1$ and $LS_2$ be the LSDSes containing $c_1$ and $c_2$, respectively.
The algorithm begins by updating $CAdj_{c_1}[id_{c_2}]$ and $CAdj_{c_2}[id_{c_1}]$.
Next, the algorithm calls $\texttt{UpdateAdj}(c_1)$ and $\texttt{UpdateAdj}(c_2)$ in order to update the $CAdj$ vectors in $LS_1$ and $LS_2$.
In case of a violation to Invariant~\ref{edgestochunk}, the algorithm executes $O(1)$ splits and merges on $c_1$ or $c_2$ together with their respective adjacent chunks, followed by $O(1)$ LSDS operations.

If $LS_1\neq LS_2$ then $u$ and $v$ are in different Euler tours, and so by Lemma~\ref{lemma:surgicaldef} a series of $O(1)$ surgical operations takes place in order to merge the two Euler tours containing $u$ and $v$ into a single Euler tour. The algorithm also adds $e$ to the dynamic tree structure of Sleator and Trajan~\cite{DBLP:journals/jcss/SleatorT83} in $O(\log n)$ worst-case time.

If $LS_1 = LS_2$, then $u$ and $v$ are in the same Euler tour. In this case, the algorithm uses the dynamic tree structure to locate the heaviest edge $e'$ on the path from $u$ to $v$ in the current MSF. Finding $e'$ takes $O(\log n)$ worst-case time.
If $w(e) < w(e')$ then the algorithm removes $e'$ from the MSF, inserts $e$ into the MSF, and updates the dynamic tree structure which costs $O(\log n)$ worst-case time.
Thus, the total worst-case time for inserting an edge is $O(J \log{J}+K+\log n)$.

\emph{Edge deletion.}
Suppose we delete edge $e=(u,v)$ from the graph. Let $c_{1}$ be the chunk containing $\pc u$ and let $c_{2}$ be the chunk containing $\pc v$. Let $LS$ be the LSDS containing $c_{1}$ and $c_{2}$.
The algorithm begins by updating $CAdj_{c_1}[id_{c_2}]$ and $CAdj_{c_2}[id_{c_1}]$ in $O(K)$ worst-case time by scanning all edges touching $c_{1}$. Next, the algorithm calls \texttt{UpdateAdj}($c_{1}$) on $LS$ to update all $CAdj$ vectors in $LS$.
If $e$ is a tree edge, then the algorithm first removes $e$ from the dynamic tree structure in $O(\log n)$ worst-case time, and then executes a series of $O(1)$ surgical operations in order to split the Euler tour containing $u$ and $v$.
Let $ET_u$ and $ET_v$ be the resulting two Euler tours containing $u$ and $v$, respectively.
Finally, the algorithm looks for a MWR edge between $ET_u$ and $ET_v$ in $O(J+K)$ worst-case time, and if such an edge $e'$ is found, then the algorithm adds $e'$ to the dynamic tree structure and executes another series of $O(1)$ surgical operations reconnecting $ET_{u}$ and $ET_{v}$.
Thus, the cost of deleting an edge is $O(J \log{J}+K+\log n)$ worst-case time.

\emph{Time cost.} Recall that $J=O(n/K)$. By setting $K = O(\sqrt{n \log n})$, the insertion and deletion costs become $O(J \log{J}+K+\log n) = O(\sqrt{n\log{n}})$ worst-case time.
\end{proof}

\section{Parallel Dynamic MSF on Sparse Graphs}
\label{sec:basic}
		
In this section we prove the following theorem.
\begin{theorem}
	\label{theorem:mstpar_sparse}
	There exists a deterministic algorithm for the dynamic MSF problem in the EREW PRAM model on sparse graphs with $m=O(n)$ edges that uses $O(\sqrt{n})$ processors and has a parallel worst-case update time of $O(\log{n})$. The resulted work of the algorithm is $O(\sqrt{n} \log{n})$.
\end{theorem}

In Section~\ref{sec:parallel_sparse} we show how to extend Theorem~\ref{theorem:mstpar_sparse} to work for general graphs, thereby proving Theorem~\ref{theorem:mstpar}.

\paragraph{The data structure.}
The algorithm uses the same data structure as described in Section~\ref{sec:seq_data_struct} with three changes.
The first change is that for each chunk $c$, the list of vertices in $c$ is augmented with a balanced 2-3 tree, denoted by $\bt c$, whose leaves are elements of the list that are in $c$. The height of $\bt c$ is $O(\log K)$.
Each vertex $v$ in $\bt{c}$ stores an \emph{edge counter} $ec_v$ which is the total number of edges incident to graph vertices whose principal copy is in the subtree of $v$; see Figure \ref{fig:chunk}.
The order of leaves in $\bt{c}$ together with an order of the at most 3 edges incident to each graph vertex whose principal copy is in $c$ defines an order on the edges touching $c$.

The second change is due to the requirements from the EREW PRAM model.
In particular, for any chunk $c$, we cannot support constant worst-case time access to the entries of $CAdj_c$ (or $Memb_c$) in parallel through a single pointer from $c$ to the array $CAdj_c$, due to the exclusive reading requirement.
Instead, we use a two dimensional matrix $C$ of size $J\times J = O(n)$ (at the end of this section we set $K=\sqrt n$) such that the entries of the $j$'th row of $C$ are exactly the entries of $CAdj_c$ where $id_c=j$. From now on, we let $CAdj_c[i]$ denote $C[id_c,i]$.
We also use the same exact method for $Memb$ arrays.

The third change, which is also due to the exclusive reading requirement, is in the $CAdj$ and $Memb$ arrays in the LSDS.
Instead of using one tree LSDS $LS$, we now use $J$ trees $S_1,S_2,\ldots,S_J$ for each LSDS, where the $j$'th tree corresponds to the chunk with id $j$.
For chunk $c$, the $id_c$'th leaf of $S_j$ contains both $CAdj_c[j]$ and $Memb_c[j]$.
We also store a pointer from $CAdj_c[j]$ and $Memb_c[j]$ to the $id_c$'th leaf of $S_j$, thereby providing direct access to that leaf.
Finally, in order to provide direct access to the root of each $S_j$, we store a matrix of size $J\times J = O(n)$ where the $(j,i)$ entry contains a pointer to the root of $S_i$ used in the $j$'th LSDS.

As in the sequential algorithm, the special case of lists containing only one chunk in the new parallel algorithm is addressed in Section~\ref{sec:onechunk}.

\paragraph{Assigning edges.} Our algorithm will often perform the task of assigning a different processor to each edge touching chunk $c$.
This assignment is implemented by a parallel operation $\getedge{c}{k}$ in which processor $p_k$ accesses the $k$'th edge incident to chunk $c$.
The operation $\getedge{c}{k}$ uses the edge counters in $\bt{c}$ together with an array $vertex$ of size $3K$ where each entry is a pointer to a vertex in $\bt c$.
We describe the implementation from the perspective of processor $p_k$ for $1\le k \le 3K$.
We emphasize that in order to implement $\getedge{c}{k}$, only $p_1$ will require access to $c$

Let $root_c$ be the root of $\bt c$.
The implementation has $h$ phases where $h=O(\log K)$ is the height of $\bt c$.
Processor $p_k$ participates in the $i$'th phase if and only if $vertex[k]\neq NULL$ at the beginning of the $i$'th phase.
Moreover, the participating processors in each phase are assigned to different vertices in $\bt c$.
In particular, processor $p_k$ is assigned to the vertex $v=vertex[k]$ with the guarantee that the rank of the rightmost edge in the subtree of $v$ is $k$.

To initialize the process, each $p_k$ sets $vertex[k]=NULL$ and if $k=1$ then $p_k$ sets $vertex[ec_{root_c}] = root_c$.
Now we begin the phases for $i=1,2,\ldots,h$.
For the $i$'th phase, if $vertex[k]=v$, then $p_k$ accesses the at most 3 children of $v$ and looks at their edge counters.
Based on these edge counters, $p_k$ computes in constant worst-case time the rank of the rightmost edge in each one of the subtrees of the children of $v$.
If the rightmost edge in the subtree of a child $u$ of $v$ is $r$, then $p_k$ sets $vertex[r]=u$.
Notice that $p_k$ necessarily sets $vertex[k]$ to be the rightmost child of $v$.
After $h$ phases all of the vertices in $vertex$ are leaves of $\bt c$, but some of the entries of $vertex$ may still be set to $NULL$.
An entry $vertex[k]=NULL$ can occur due to one of two reasons: either $k$ is larger than the number of edges touching $c$, or the principal copy of the edge that $p_k$ is accessing is also the principal copy of another edge which is being accessed by a different processor.
However, due to the invariant that the rank of the rightmost edge in the subtree of $vertex[k]$ is $k$ and the fact that the maximum degree in the graph is 3, the principal copy that $p_k$ is looking for is either in $vertex[k+1]$ or $vertex[k+2]$.
Thus, within 3 more steps, $p_k$ is able to access the principal copy and complete the task.
Thus, the operation $\getedge{c}{k}$ costs $O(\log{K})$ worst-case time.

\begin{figure}
	\includegraphics[scale=0.8, trim=3cm 14cm 4cm 3.3cm]{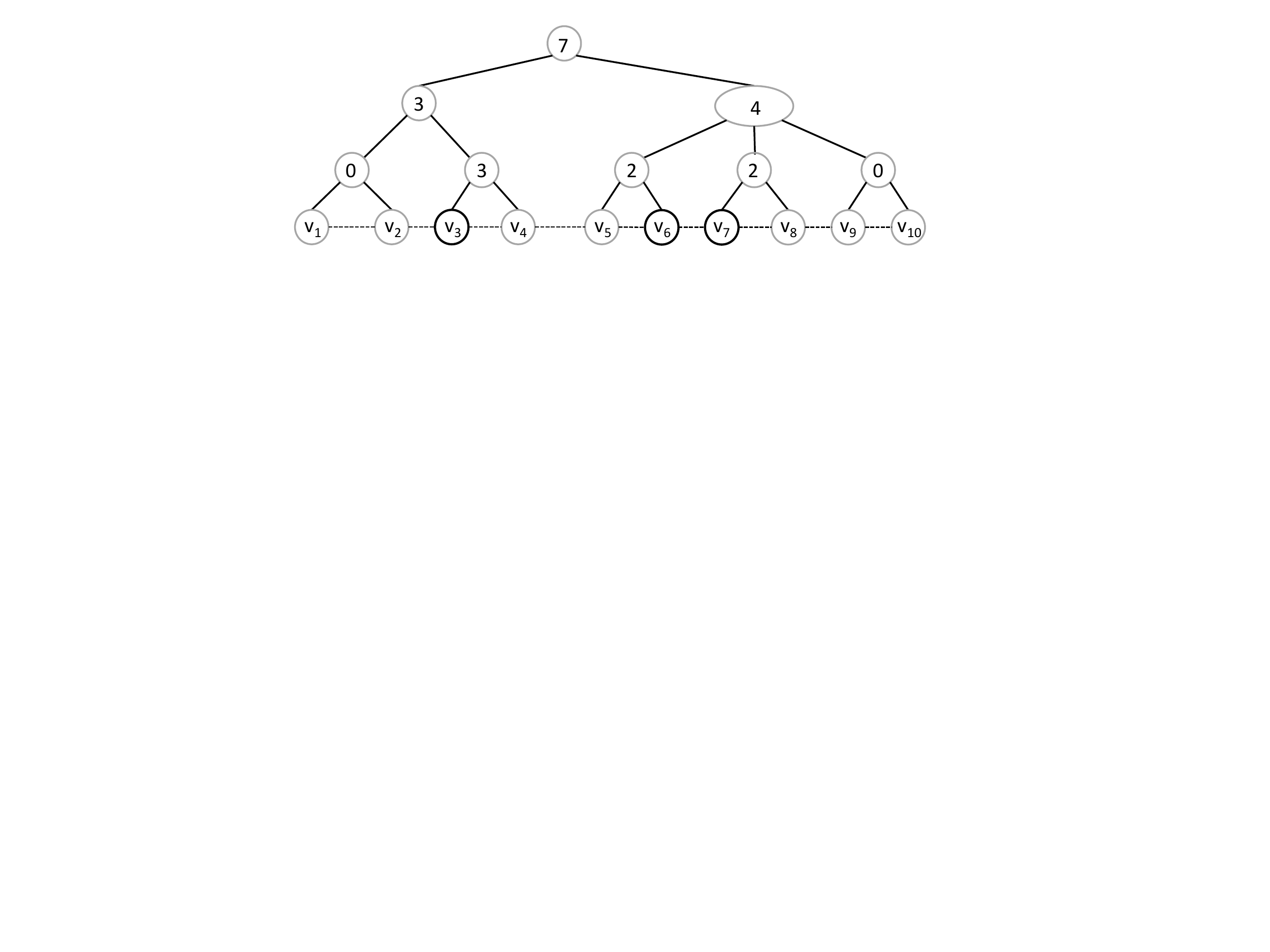}
	\caption{\small The 2-3 tree $\bt c$ built on the list of vertices inside chunk $c$. $v_{3}$, $v_{6}$ and $v_{7}$ are the principal copies inside chunk $c$, and the number of edges touching $v_{3}$, $v_{6}$ and $v_{7}$ is 3, 2, and 3, respectively.
The numbers inside the inner tree vertices are the edge counters.}
	\label{fig:chunk}
\end{figure}

\subsection{Splitting and Merging Chunks}
\begin{lemma}
  	\label{lem:chunksplitmergepar}
  	There exists an algorithm in the EREW PRAM model that supports splits and merges of chunks such that each operation costs $O(\log K)$ parallel worst-case time, using $O(J+K)$ processors.
\end{lemma}
\noindent\begin{proof}
\emph{Splitting.}
Recall that splitting a chunk $c$ can happen for one of two reasons: (1) either the list containing $c$ needs to be split at a given vertex $u$ that is in $c$, or (2) $n_c>3K$ thereby violating Invariant~\ref{edgestochunk}. In the second case, processor $p_1$ locates the split location in $O(\log K)$ worst-case time by traversing down $\bt c$ using the edge counters.
Thus, we assume from now that the algorithm knows the split location $u$.

Processor $p_1$ splits $\bt c$ at vertex $u$ in $O(\log K)$ worst-case time.
Let $c_{1}$ and $c_{2}$ be the resulting chunks where $c_1$ contains the first part of $\bt c$ and $c_2$ contains the second part.
Processor $p_1$ sets $id_{c_{1}} = id_{c}$ and allocates a new (unique) id for $c_{2}$.
We now focus on creating $CAdj_{c_1}$, since $CAdj_{c_2}$ is created in the same manner.

The sequential algorithm for constructing $CAdj_{c_1}$ (in the proof of Theorem~\ref{theorem:mstseq}) scans all of the $O(K)$ edges touching $c_1$.
In the parallel setting, accessing all of the edges in parallel does not suffice since there could be several edges touching both $c_1$ and $\hat c$ for some other chunk $\hat c$, and the algorithm needs to store only the minimum weight of such an edge. To solve this issue we do the following.

The algorithm uses $J$ balanced binary tournament trees $T_1,T_2,\ldots, T_J$, where each tree has $3K$ leaves.
Each vertex $z$ in $T_j$ stores a value $A_z$ initialized to $\infty$\footnote{Reusing and initializing a temporary data structure in the parallel setting is implemented by either using a timestamp for each word of memory or rolling back all of the memory changes after the operation completes, thereby allowing the cost analysis to ignore the initialization cost.}.
The algorithm uses an iterative process for implementing a special tournament-like process.
During the iterative process, each processor will initially be \emph{active} until the processor decides to become \emph{inactive} and no longer participates in the process.
The iterative process implicitly uses an \emph{exclusive-assignment property} which states that each participating processor is assigned to a vertex in some tree $T_j$ such that there are no two processors that are assigned to the same vertex.
At the beginning of the $i$'th iteration the active processors are assigned to vertices whose height is $i-1$.

The initialization of the iterative process is as follows.
For each $1 \leq k \leq 3K$, processor $p_{k}$ sets itself as active and executes $\getedge{c_{1}}{k}$ thereby gaining access to $e_{k} = (u_{k}, v_{k})$, which is the $k$'th edge adjacent to chunk $c_{1}$.
Assume without loss of generality that $\pc {u_k} \in c_1$.
Let $c_{v_k}$ be the chunk such that $\pc{v_k} \in c_{v_k}$, and denote $id = id_{c_{v_k}}$.
Processor $p_k$ assigns itself in $O(1)$ worst-case time (using a lookup table) to the $k$'th leaf of $T_{id}$, denoted by $\ell_k$, and sets $A_{\ell_k}= w(e_k)$.
Thus, the exclusive-assignment property holds.

Each iteration has four \emph{synchronous} phases.
Recall that only active processors continue to participate in the process.
\begin{itemize}\compactify
  \item \emph{Phase 1.} If $p_k$ is assigned to a vertex $z$ that is the \emph{left} child of its parent $parent(z)$ then $p_k$ sets $A_{parent(z)}= w(e_k)$.
  \item \emph{Phase 2.} If $p_k$ is assigned to a vertex $z$ that is the \emph{right} child of its parent $parent(z)$ then: if $A_{parent(z)}> w(e_k)$ then $p_k$ sets $A_{parent(z)} = w(e_k)$ and otherwise $p_k$ becomes inactive.
  \item \emph{Phase 3.} If $p_k$ participated in the first phase and $A_{parent(z)} < w(e_k)$ then $p_k$ becomes inactive.
  \item \emph{Phase 4.} If $parent(z)$ is the root of a tournament tree, then the iterative process ends. Otherwise, $p_k$ is assigned to $parent(z)$.
\end{itemize}
Notice that by the exclusive-assignment property we are guaranteed that during the first two phases no two processors are writing to the same location in memory at the same time.
Also, we are guaranteed that if two processors are assigned to sibling vertices, then after the third phase the processor that is assigned to the lighter edge remains active, with ties favoring the left vertex, while the other processor becomes inactive.
Thus, after the third phase, if $p_k$ is still active then there is no other active processor $p_{\hat k}$ that is currently assigned to $z$, and so after the fourth phase the exclusive-assignment property holds.
At the end of the iterative process, the processor that is at the root of $T_j$ sets $CAdj_{c_1}[j]=A_{root(T_j)}$.

Finally, for each chunk $c'$, the algorithm sets $CAdj_{c'}[id_{c_1}] = CAdj_{c_1}[id_{c'}]$ and sets  $CAdj_{c'}[id_{c_2}] = CAdj_{c_2}[id_{c'}]$, which takes $O(1)$ parallel worst-case time using $O(J)$ processors.
Thus, the cost for splitting a chunk is $O(\log K)$ parallel worst-case time, using $O(J+K)$ processors.

\emph{Merging.}
Processor $p_1$ merges $\bt{c_{1}}$ and $\bt{c_{2}}$ in $O(\log K)$ worst-case time. Let $c$ denote the resulting chunk containing the concatenation of the lists of vertices represented by $\bt{c_{1}}$ and $\bt{c_{2}}$.
Processor $p_1$ sets $id_{c} = id_{c_{1}}$.
The new $CAdj$ array for $c$ is created by performing an entry-wise minimum of $CAdj_{c_1}$ and $CAdj_{c_2}$ in $O(1)$ parallel worst-case time using $O(J)$ processors.
Finally, for each chunk $c'$, the algorithm sets $CAdj_{c'}[id_{c}] = CAdj_{c}[id_{c'}]$ and sets $CAdj_{c'}[id_{c_2}] = \infty$, which takes $O(1)$ parallel worst-case time using $O(J)$ processors.
Thus, the cost for merging a chunk is $O(\log K)$ parallel worst-case time, using $O(J)$ processors.
\end{proof}

\subsection{LSDS Operations}
\begin{lemma}
\label{lem:lsdspar}
	There exists an implementation of the LSDS in the EREW PRAM model using $O(J)$ processors where each of the operations \texttt{LSInsert}, \texttt{LSDelete}, \texttt{LSJoin}, \texttt{LSSplit} and \texttt{UpdateAdj} takes $O(\log{J})$ parallel worst-case time.
\end{lemma}
\begin{proof}
\emph{All operations except for \texttt{UpdateAdj}.}	
	Recall that in the proof of Lemma~\ref{lem:lsdsseq} each basic tree operation touches at most $O(\log{J})$ vertices in the tree. We use a similar implementation as in Lemma~\ref{lem:lsdsseq}, but now processor $p_j$ for $1 \leq j \leq J$ performs the basic tree operations on $S_j$.
Thus, the total parallel worst-case time for each operation except for UpdateAdj is $O(\log{J})$.

\emph{Operation \texttt{UpdateAdj($c$)}.}
We again use a similar implementation as in Lemma~\ref{lem:lsdsseq}, with the following changes.
For each $S_j$, updating the path from the leaf representing $c$ to the root of $S_j$ costs $O(\log{J})$ parallel worst-case time using $O(J)$ processors.
In order to update $S_{id_c}$, processor $p_{j}$ for $1 \leq j \leq J$ is responsible for handling the leaf representing $chunks[j]$, which is accessible through the pointer stored in $CAdj_{chunks[j]}[id_c]$. We now need to sweep up $S_{id_c}$ in parallel, starting from all of the leaves of $S_{id_c}$. This process is described next.

The algorithm begins an iterative process where at the beginning of the $i$'th iteration there is a unique processer assigned to each vertex of height $i-1$ in $S_{id_c}$.
At the $i$'th iteration, suppose $p_{j}$ is assigned to vertex $z$ of height $i-1$. Then $p_j$ is reassigned to $parent(z)$ only if $z$ is the \emph{leftmost} child of $parent(z)$.
If $z$ is \emph{not} the leftmost child of $parent(z)$, then $p_{j}$ halts.
Thus, each vertex at height $i$ is assigned to exactly one processor.
If $p_{j}$ did not halt then $p_j$ updates the value stored in $parent(z)$ in $O(1)$ worst-case time.
The iterative process ends at the root, which happens after $O(\log{J})$ steps.
The parallel worst-case time cost \emph{per level} is $O(1)$ and $O(\log{J})$ worst-case time for the entire procedure. The number of processors used is $O(J)$.
\end{proof}

\subsection{Surgical Operations}
	\begin{lemma}	
		\label{lem:surgicalimplpar}
		There exists an algorithm in the EREW PRAM model in which each surgical operation on lists costs $O(\log J + \log K)$ parallel worst-case time using $O(J+K)$ processors and finding a MWR edge costs $O(\log J + \log K)$ parallel worst-case time using $O(J+K)$ processors.
	\end{lemma}
	
	\noindent\begin{proof}
The implementation of both splitting and merging lists remains the same as in the proof of Lemma~\ref{lem:surgicalimpl}, but this time applying Lemma~\ref{lem:chunksplitmergepar} instead of Lemma~\ref{lem:chunksplitmerge}.
So the operation of splitting a list costs $O(\log K )$ parallel worst-case time using $O(J+K)$ processors, and the operation of merging two lists costs $O(\log J + \log K)$ parallel worst-case time using $O(J+K)$ processors.
	
	\emph{Finding a MWR edge.}
	The algorithm constructs the array $\gamma$, as defined in the proof of Lemma~\ref{lem:surgicalimpl}, but now processor $p_j$ for $j\in [J]$ computes $\gamma[j]$ in $O(1)$ parallel worst-case time by accessing the root of $S_j$ in each LSDS in constant worst-case time (using the lookup matrix).
	Let ${\hat c}=\arg \min_{\text{chunk } c}\{\gamma[id_c]\}$.
    Recall that the minimum weight edge between $ET_1$ and $ET_2$ (as defined in the proof of Lemma~\ref{lem:surgicalimpl}) touches a vertex $u$ such that $\pc{u} \in \hat c$.
    The algorithm in the proof of Lemma~\ref{lem:surgicalimpl} computes $id_{\hat c}$ in $O(J)$ worst-case time by scanning $\gamma$ and finding the smallest entry.
    In the EREW PRAM model, the algorithm uses a tournament tree to find the smallest entry, which costs $O(\log{J})$ parallel worst-case time using $O(J)$ processors.
	Next, processor $p_k$ for $k\in [3K]$ accesses edge $e_k=getEdge_{\hat c}(k)$.
    Let $e_k=(u,v)$ where $\pc u \in \hat c$.
    In the CREW PRAM model, processor $p_k$ verifies in $O(1)$ whether the chunk $c_v$ containing $\pc v$ is in $LS_1$ (as defined in the proof of Lemma~\ref{lem:surgicalimpl}) by looking at the $Memb$ value in the root of $S_{id_{c_v}}$ of $LS_1$. Using the reduction of~\cite{Jaja}, this process costs $O(\log K)$ worst-case time in the EREW model.
    Finally, the algorithm picks the lightest edge via a tournament tree algorithm whose participants are the processors whose edge passed the verification.
	The algorithm for finding the MWR edge takes $O(\log J +\log K)$ parallel worst-case time, using $O(J)$ processors.
	\end{proof}

\subsection{Graph Updates}
\begin{proof}[Proof of Theorem~\ref{theorem:mstpar}]
\emph{Edge insertion.}
The algorithm for inserting an edge is the same as in the sequential algorithm in the proof of Theorem \ref{theorem:mstseq}, but this time applying Lemmas~\ref{lem:chunksplitmergepar},~\ref{lem:lsdspar}, and~\ref{lem:surgicalimplpar} instead of Lemmas~\ref{lem:chunksplitmerge},~\ref{lem:lsdsseq}, and~\ref{lem:surgicalimpl}. The parallel worst-case update time is $O(\log{J} + \log{K})$, by using $O(J+K)$ processors.
		
\emph{Edge deletion.}
The algorithm for deleting an edge is the same as in the sequential algorithm in the proof of Theorem \ref{theorem:mstseq}, except for two changes:
(1) the edge deletion algorithm applys Lemmas~\ref{lem:chunksplitmergepar},~\ref{lem:lsdspar}, and~\ref{lem:surgicalimplpar} instead of Lemmas~\ref{lem:chunksplitmerge},~\ref{lem:lsdsseq}, and~\ref{lem:surgicalimpl}, and
(2) the new minimum weight edge connecting chunks $c_1$ and $c_2$ (as defined in the edge deletion operation in the proof of Theorem~\ref{theorem:mstseq}) is found in the EREW PRAM model by using a tournament tree which costs $O(\log{K})$ parallel worst-case time using $O(K)$ processors.
The parallel worst-case update time is $O(\log{J} + \log{K})$, by using $O(J+K)$ processors.
	
\emph{Time cost.} By setting $K = O(\sqrt{n})$, the insertion and deletion costs become $O(\log{n})$ parallel worst-case time using $O(\sqrt n)$ processors, for a total work of $O(\sqrt {n} \log n)$.
\end{proof}

\section{Conclusion}
We described an algorithm for solving dynamic MSF on sparse graphs in the EREW PRAM model that uses $O(\sqrt n)$ processors and has $O(\log n)$ worst-case update time. The resulted work of the algorithm is $O(\sqrt{n} \log{n})$.
By extending the sparsification technique to work in the EREW PRAM model (see Section~\ref{sec:parallel_sparse}), the algorithm can be used for solving dynamic MSF on general graphs with the same complexities.
Thus, the total work is $O(\sqrt n \log n)$. We leave open the task of designing a solution that has a parallel $O(\log n)$ worst-case update time, but only $O(\sqrt n)$ work, thereby matching the amount of work used in the sequential solutions.

\section{Sparsification in the EREW PRAM Model}\label{sec:parallel_sparse}

\subsection{The Sparsification Tree}
We begin by following the construction of~\cite{EppsteinGIN97}.
The construction of the sparsification tree structure begins with a recursive partitioning of the vertices of the graph into two evenly split halves.
We end up with a complete binary tree called the \emph{vertex-partition tree} in which a tree vertex at level $i$ is associated with $\frac n {2^i}$ graph vertices.
The vertex-partition tree is used to partition the edges of the graph into an \emph{edge-partition tree} as follows.
For every unordered pair of vertex-partition tree vertices $\alpha$ and $\beta$ at level $i$ (including the pair in which $\alpha=\beta$) with corresponding graph vertex sets $V_\alpha$ and $V_\beta$, we create an edge-partition tree vertex $E_{\alpha \beta}$ in the edge-partition tree.
The vertex $E_{\alpha \beta}$ conceptually corresponds to the set of all edges between vertices from $V_\alpha$ and vertices from $V_\beta$.
If the vertex-partition tree partitions $V_\alpha$ ($V_{\beta}$) into $V_{\alpha_1}$ and $V_{\alpha_2}$ ($V_{\beta_1}$ and $V_{\beta_2}$) then the children of $E_{\alpha \beta}$ in the edge-partition tree are $E_{\alpha_1 \beta_1}, E_{\alpha_1 \beta_2}, E_{\alpha_2 \beta_1},$ and $E_{\alpha_2 \beta_2}$.
Notice that if $\alpha = \beta$ then $E_{\alpha_1 \beta_2}$ and $E_{\alpha_2 \beta_1}$ are the same.
Thus, if $E_{\alpha \beta}$ is not a leaf then $E_{\alpha \beta}$ has $3$ or $4$ children, depending on whether $\alpha=\beta$ or not.

Each $E_{\alpha \beta}$ maintains a \emph{local graph} ${G_{\alpha \beta}}\subseteq G$ whose set of edges is the union of the MSF edges of the children of $E_{\alpha \beta}$. Thus, the size of a graph at level $i$ is $O(\frac n {2^i})$.
Each $E_{\alpha \beta}$ maintains an instance of dynamic MSF on ${G_{\alpha \beta}}$.
Eppstein et al.~\cite{EppsteinGIN97} proved that the MSF at the root of the edge-partition tree is the MSF of the graph $G$.

For $u\in V_\alpha$ let $u_{\alpha \beta}$ be the copy of graph vertex $u\in V$ in ${G_{\alpha \beta}}$.
Similarly, for $u\in V_\beta$ let $u_{\beta \alpha}$ be the copy of graph vertex $u\in V$ in ${G_{\alpha \beta}}$.
Notice that, by the construction of the edge-partition tree, \textbf{if}: (1) $E_{\alpha \beta}$ is not a leaf, (2) the vertex-partition tree partitions $V_\alpha$ ($V_{\beta}$) into $V_{\alpha_1}$ and $V_{\alpha_2}$ ($V_{\beta_1}$ and $V_{\beta_2}$), and (3) $u\in V_{\alpha_1}$, \textbf{then} $u$ has copies in both ${G_{\alpha_1 \beta_1}}$ and ${G_{\alpha_1 \beta_2}}$.

Let $e_{\alpha \beta}$ be the copy of graph edge $e\in E$ in ${G_{\alpha \beta}}$, if it exists.
Moreover, if $e_{\alpha \beta}$ is a tree edge for the MSF of $G_{\alpha \beta}$ then $e_{\alpha' \beta'} \in {G_{\alpha' \beta'}}$ where $E_{\alpha' \beta'}$ is the parent of $E_{\alpha \beta}$.

\paragraph{Pointers between copies.}
Notice that the dynamic MSF data structure is a data structure on edges of graphs (even if the runtime depends on the number of vertices), and so the data structure does not explicitly store singleton vertices of ${G_{\alpha \beta}}$.
Suppose $E_{\alpha \beta}$ is not a leaf and suppose $u_{\alpha \beta}$ is not a singleton vertex.
Let $E_{\alpha_1 \beta_1}$ and $E_{\alpha_2 \beta_2}$ be the two children of $E_{\alpha \beta}$ that contain $u_{\alpha_1 \beta_1}$ and $u_{\alpha_2 \beta_2}$, respectively.
Then $u_{\alpha \beta}$ stores \emph{vertex-copy pointers} to both $u_{\alpha_1 \beta_1}$ and $u_{\alpha_2 \beta_2}$.
Moreover, every non singleton graph vertex $u\in V$ stores two vertex-copy pointers to the two copies of $u$ in the root of the edge-partition tree\footnote{Notice that each graph vertex appears twice in the root of the edge-partition tree, since the root corresponds to all edges in $V\times V$.}.

Suppose $E_{\alpha \beta}$ is not the root of the edge-partition tree and suppose $e_{\alpha \beta}$ is a tree edge in the MSF of ${G_{\alpha \beta}}$.
Let $E_{\alpha' \beta'}$ be the parent of $E_{\alpha \beta}$. Then $e_{\alpha \beta}$ stores a bidirectional  \emph{edge-copy pointer} to $e_{\alpha' \beta'}$.
Notice that, by construction, the leaves in the edge-partition tree have a bijection with pairs of vertices from $G$.
Thus, each graph edge $e=(u,v)\in E$ stores an edge-copy pointer to the copy of $e$ in ${G_{\alpha \beta}}$ where $V_\alpha =\{u\}$ and $V_\beta =\{v\}$.

Following Eppstein et al.~\cite{EppsteinGIN97}, we modify the edge-partition tree in order to reduce the space usage.
The data structure stores $E_{\alpha \beta}$ only if there is at least one edge between a vertex in $V_\alpha$ and a vertex in $V_\beta$.
Thus the total number of stored leaves is $m$ and since we do not store singleton vertices, the total space usage becomes $O(m\log n)$.
The modified edge-partition tree is the sparsification tree which we denote by $T$.

\subsection{Sequential Sparsification}
We describe the sequential sparsification in a particular way that caters towards the parallel implementation.

\subsubsection{Edge Insertion}
Suppose we insert edge $e=(u,v)$ to $G$.
Starting at the root of $T$ the algorithm traverses down $T$ with the goal of visiting all the vertices $E_{\alpha \beta}$ in $T$ such that $u\in V_\alpha$ and $v\in V_\beta$.
This traversal takes place by moving from  $E_{\alpha \beta}$ to its only child $E_{\alpha_1 \beta_1}$, if such a child explicitly exists, such that $u\in V_{\alpha_1}$ and $v\in V_{\beta_1}$.
Once such a child does not exist, the algorithm completes the path towards the leaf corresponding to $e$ by adding the missing vertices to $T$.

Next, the algorithm once again traverses the path from the root of $T$ down to the leaf corresponding to $e$, together with the vertex-copy pointers, and whenever the algorithm visits a graph ${G_{\alpha \beta}}$ that does not contain either $u_{\alpha \beta}$ or $v_{\beta \alpha}$, the algorithm adds the missing $u_{\alpha \beta}$ or $v_{\beta \alpha}$ to ${G_{\alpha \beta}}$.
As this traversal takes place, the algorithm stores a list of the copies of $u$ and $v$ in an array $VCopy$. In particular, if $E_{\alpha \beta}$ is at level $i$ in $T$ such that $u\in V_\alpha$ and $v\in V_\beta$ then $VCopy[i] = (u_{\alpha \beta},v_{\beta \alpha})$.

For each $E_{\alpha \beta}$ on the path with parent $E_{\alpha' \beta'}$, the algorithm uses the dynamic tree data structure of Sleator and Tarjan~\cite{DBLP:journals/jcss/SleatorT83} to test whether $e_{\alpha \beta}$ should be added to the MSF of ${G_{\alpha \beta}}$.
For efficiency purposes, this test uses the direct access to the copies of $u$ and $v$ that is given by $VCopy$. If the answer is yes, then the algorithm inserts $e_{\alpha' \beta'}$ into ${G_{\alpha' \beta'}}$ while also updating the dynamic MSF data structure of ${G_{\alpha' \beta'}}$.
The insertion of $e_{\alpha \beta}$ into ${G_{\alpha \beta}}$, if needed, is initiated by the same test that takes place at the appropriate child of $E_{\alpha \beta}$.
Notice that adding $e_{\alpha \beta}$ to the MSF of ${G_{\alpha \beta}}$ may cause a different edge $e'_{\alpha \beta}$ to be removed from the same MSF. In such a case, $e'_{\alpha' \beta'}$ is deleted from  ${G_{\alpha' \beta'}}$. Finally, the algorithm updates the appropriate vertex-copy and edge-copy pointers in a straightforward manner.

\paragraph{Cost analysis.} Adding the missing tree vertices to $T$ and constructing $VCopy$ costs $O(\log n)$ worst-case time.
For each level $i$ the algorithm performs a test using the dynamic tree data structure, and then executes a constant number of graph updates on a graph of size $O(\frac n{2^i})$. Thus, the total worst-case time cost for all levels is $\sum_{i=1}{O(\log n)} O(\sqrt{\frac n {2^i}\log \frac n {2^i}}) = O(\sqrt {n\log n})$.

\subsubsection{Edge Deletion}
Suppose we delete edge $e=(u,v)$ from $G$.
The algorithm begins by traversing up $T$ using the edge-copy pointers starting from $e$, until the algorithm reaches the highest vertex $E_{\alpha^* \beta^*}$ such that $e_{\alpha^* \beta^*}$ is in $G_{\alpha^* \beta^*}$.
For each vertex $E_{\alpha \beta}$ on the path from the leaf corresponding to $e$ and $E_{\alpha^* \beta^*}$, the algorithm removes $e_{\alpha \beta}$ from ${G_{\alpha \beta}}$.
If $e_{\alpha \beta}$ was the only edge in ${G_{\alpha \beta}}$ then $E_{\alpha \beta}$ is removed from $T$.
If the removal of $e_{\alpha \beta}$ causes a copy of a graph vertex to become a singleton, then the copy is removed from $G_{\alpha \beta}$.

The algorithm makes use of an array $REdges$ of size $O(\log n)$, with one entry per level in $T$.
If $E_{\alpha \beta}$ is at level $i$ in $T$ and $e_{\alpha \beta}$ was a tree edge then removing $e_{\alpha \beta}$ may cause a different edge $e'_{\alpha \beta}$ to become a tree edge.
In such a case, we set $REdges[i]= e'_{\alpha \beta}$.
Otherwise we set $REdges[i]= NULL$.
For $E_{\alpha \beta}$ at level $i$, the lightest edge from $\{REdges[1],REdges[2],\ldots,REdges[i-1]\}$ is inserted into $G_{\alpha \beta}$.
Determining which edge copy to insert at each level costs $O(\log n)$ worst-case time by scanning $REdges$.
Finally, the algorithm updates edge-copy and vertex-copy pointers as necessary.

\paragraph{Cost analysis.} Similar to the insertion cost, the cost of a deletion is $\sum_{i=1}{O(\log n)} O(\sqrt{\frac n {2^i}\log \frac n {2^i}}) = O(\sqrt {n\log n})$ worst-case time.

\subsection{Parallel Sparsification}
Notice that the operations in the sequential sparsification that take place during graph updates can be classified into two classes. The first class are operations do not benefit from parallelization, which include the first two traversals during the insertion of an edge (including the constructing of $VCopy$), accessing all of the copies of a deleted edge, and using $REdges$ to determine which edges need to be inserted into local graphs. All of these operations cost $O(\log n)$ sequential worst-case time. The second class are operations that do benefit from parallelization,
since these operations can be executed independently on each level in $T$. These include determining whether a new edge will become a tree edge in a local graph, a constant number of insertions and deletions into a local graph, and the construction of $REdges$. By applying Theorem~\ref{theorem:mstpar} to each dynamic MSF data structure, the total worst-case time cost of each of these operations is $O(\log n)$ while the number of processors used at level $i$ is $O(\sqrt{\frac n {2^i}})$. Thus the total worst-case time cost is $O(\log n)$ while the total number of processors used is $O(\sqrt n)$ for a total of $O(\sqrt n \log n)$ work.

\section{Lists Containing Only One Chunk}
\label{sec:onechunk}
In the special case of a list containing only one chunk $c$ we have $n_c<K$.
We call such a list a \emph{short list}, and the algorithm does not give a unique id to the only chunk of the list.
Moreover,
The algorithm does not maintain a $CAdj$ vector for this chunk.

\paragraph{Joining lists.} Suppose the algorithm joins two lists $L_1$ and $L_2$, and $L_2$ is short. Let $c_2$ be the single chunk in $L_2$. If the concatenation $L_{1}L_{2}$ is \emph{not} short, then $c_2$ is given a unique id from $[J]$, and a new LSDS representing the concatenation is constructed. Next, $c_2$ is merged and split with the adjacent chunk in order to restore Invariant~\ref{edgestochunk}.

\paragraph{Splitting lists.} Suppose the algorithm splits a list $L$ into two lists $L_1$ and $L_2$. If $L_2$ is short with a single chunk $c_2$, then does not allocate a new id to $c_2$.

\paragraph{Finding a MWR edge.} When trying to find a MWR edge between two lists and at least one of the lists is short, the algorithm scans all vertices in the short list in order to find the minimum replacement edge in $O(K)$ worst-case time (or in $O(\log K)$ parallel worst-case time using $O(K)$ processors in the EREW PRAM model, by using a tournament tree).

\section{Acknowledgments}
This work is supported in part by ISF grant 1278/16. This project has received funding from the European Research Council (ERC) under the European Union’s Horizon 2020 research and innovation programme (grant agreement No 683064).

\bibliography{DynConnThes}
\bibliographystyle{plain}

\end{document}